\documentclass[sigconf, twocolumn]{acmart}
\pdfoutput=1
\AtBeginDocument{%
  \providecommand\BibTeX{{%
    \normalfont B\kern-0.5em{\scshape i\kern-0.25em b}\kern-0.8em\TeX}}}
\setcopyright{acmcopyright}
\copyrightyear{2021}
\acmYear{2021}
\acmConference[ISSAC '21]{International Symposium on Symbolic and Algebraic
Computation}{July 18--22, 2021}{Saint Petersburg, Russia}
\acmBooktitle{International Symposium on Symbolic and Algebraic Computation
(ISSAC '21), July 18--22, 2021, Saint Petersburg, 
}
\acmPrice{15.00}
\acmDOI{10.1145/...}
\acmISBN{978-1-4503-6084-5/19/07}

\usepackage[utf8]{inputenc}
\usepackage{pgf,tikz,pgfplots}
\pgfplotsset{compat=1.15}
\usepackage{mathrsfs}
\usepackage{multirow}
\usetikzlibrary{arrows}
\usepackage{fancyhdr}
\usepackage{algorithm}
\usepackage{algorithmic}

\usepackage{subfigure}
\usepackage{booktabs}
\newcommand{\Z}{\mathbb{Z}}

\newcommand{\R}{\mathbb{R}}
\newcommand{\K}{\mathbb{K}}
\newcommand{\var}{\mathop{\mathtt{var}}}
\newcommand{\res}{\mathop{\mathtt{res}}}
\newcommand{\disc}{\mathop{\mathtt{dis}}}
\newcommand{\degree}{\mathop{\mathtt{deg}}}

\newcommand{\proj}{\mathtt{Proj}}
\newcommand{\projmc}{\proj_{\text{mc}}}
\newcommand{\projbrown}{\proj_{\text{br}}}
\newcommand{\coeff}{\mathop{\mathtt{coeff}}}
\newcommand{\lc}{\mathop{\mathtt{lc}}}
\newcommand{\cont}{\mathop{\mathtt{cont}}}
\newcommand{\child}{\mathop{\mathtt{child}}}

\newtheorem{remark}{Remark}

\allowdisplaybreaks

\begin{document}

\title{Choosing the Variable Ordering for Cylindrical Algebraic Decomposition via Exploiting Chordal Structure}
\author{Haokun Li, Bican Xia, Huiying Zhang, and Tao Zheng}
\email{haokunli, huiyingz, 1601110051@pku.edu.cn} 
\email{xbc@math.pku.edu.cn}
\affiliation{%
  \institution{School of Mathematical Sciences, Peking University}
  \state{Beijing}
  \country{China}}






\begin{abstract}


Cylindrical algebraic decomposition (CAD) plays an important role in the field of real algebraic geometry and many other areas. As is well-known, the choice of variable ordering while computing CAD has a great effect on the time and memory use of the computation as well as the number of sample points computed. In this paper, we indicate that typical CAD algorithms, if executed with respect to a special kind of variable orderings (called ``the perfect elimination orderings''), naturally preserve chordality, which is an important property on sparsity of variables. Experimentation suggests that if the associated graph of the polynomial system in question is chordal (\emph{resp.}, is nearly chordal), then a perfect elimination ordering of the associated graph (\emph{resp.}, of a minimal chordal completion of the associated graph) can be a good variable ordering for the CAD computation. That is, by using the perfect elimination orderings, the CAD computation may produce a much smaller full set of projection polynomials than by using other naive variable orderings. More importantly, for the complexity analysis of the CAD computation via a perfect elimination ordering, an $(m,d)$-property of the full set of projection polynomials obtained via such an ordering is given, through which the ``size'' of this set is characterized. This property indicates that when the corresponding perfect elimination tree has a lower height, the full set of projection polynomials also tends to have a smaller ``size''. This is well consistent with the experimental results, hence the perfect elimination orderings with lower elimination tree height are further recommended to be used in the CAD projection.
\end{abstract}

\maketitle

\section{Introduction}

Cylindrical algebraic decomposition (CAD) has been widely used in real algebraic geometry and beyond 
since it was introduced by Collins in 1975 \cite{DBLP:conf/automata/Collins75}. A CAD of an Euclidean space $\mathbb{R}^n$ for a given polynomial set $P\subset\mathbb{Z}[x_1,\ldots,x_n]$ is a cylindrical partition of the space $\mathbb{R}^n$ so that every polynomial in the set $P$ is sign-invariant in each component (called ``cell'') of this partition. There are many CAD algorithms and their variants, see for example   \cite{DBLP:conf/automata/Collins75,DBLP:journals/jsc/McCallum88,hong1990improvements,Hong,collins1991partial,McCallum2,DBLP:journals/jsc/Brown01a,triCAD,hong2012variant,brown2015open,han2016proving,strzebonski2016cylindrical,DBLP:journals/jsc/BradfordDEMW16,han2017open}. 
A typical CAD algorithm usually contains two parts: the \emph{projection} part and the \emph{lifting} part. In the projection part, one eliminates the variables of the polynomial set $P$ successively in some order. Then the sample points obtained in some one-dimensional space via the full set of projection polynomials are lifted step-by-step to sample points in the original ($n$-dimensional) space during the lifting part and each sample point in $\mathbb{R}^n$ represents a cell of the CAD. 

It is long known that the variable ordering used in the projection part has a great effect on the time and memory use of the CAD computation as well as the number of sample points computed. There are many researches focusing on the problem of choosing a good variable ordering for the CAD computation. For instance, in \cite{brown2004companion,dolzmann2004efficient,bradford2013optimising}, various heuristics are provided to suggest a relatively good variable ordering based on different rules. Latter in \cite{huang2014applying,england2019comparing,chen2020variable}, the methods based on machine learning and artificial neural network are used to choose a good variable ordering for CAD computation. 

The present paper is inspired by the previous work \cite{cifuentes2017chordal,mou2019chordal,mou2019chordality,chen2020chordality} on combining the \emph{chordal structure} with triangular decomposition of polynomial sets. In \cite{cifuentes2017chordal}, based on the computation of triangular decomposition, Cifuentes and Parrilo compute the chordal network of a polynomial set by exploiting the chordal structure of its associated graph. Latter in \cite{mou2019chordal}, Mou \emph{et al} indicate that Wang's algorithm and a subresultant-based algorithm, both in top-down style, for computing the triangular decomposition of a polynomial set preserve the chordal structure. And another subresultant-based  algorithm for computing the regular decomposition in the same style is also proved to preserve the chordal structure. Then in \cite{mou2019chordality} by Mou and Lai, it is proved that Wang's algorithm for computing the simple decomposition of a polynomial set in top-down style preserves the chordal structure of the polynomial set as well. Recently in \cite{chen2020chordality}, Chen proves that an incremental algorithm for computing the triangular decomposition of a polynomial set preserves the chordal structure, too. On the contrary, Cifuentes and Parrilo indicate that the computation of Gr{\"o}bner basis of a polynomial set seems to violate the chordal structure \cite{cifuentes2016exploiting}.

In this paper we take advantage of the chordal structure of the associated graph of a polynomial set while computing CAD for this set to exploit its variable sparsity. We first indicate that all the basic operations usually used in the CAD projection algorithms, such as resultant, discriminant, subresultant, \emph{etc}, preserve chordal structures. Based on this, it is an easy corollary that typical CAD algorithms naturally preserve chordal structures (if executed with respect to a perfect elimination ordering), thus also the variable sparsity pattern embedded in the chordality. Some experimental results show that a perfect elimination ordering of a chordal structure of the polynomial set in question may be a good variable ordering for computing CAD compared with other naive ones. More importantly, for the complexity analysis of the CAD computation via a perfect elimination ordering, an $(m,d)$-property for the full set of projection polynomials obtained by using such an ordering is provided, through which the ``size'' of this set is characterized. According to that property, when the corresponding perfect elimination tree has a lower height, the full set of projection polynomials tends to have a smaller “size”. This is further supported by some experimental results. And because of that, the perfect elimination orderings with lower elimination  tree  height  are  further  recommended while computing the  CAD projection.

The rest of this paper is organized as follows: Section \ref{section:pre} includes some basic definitions. In Sections \ref{section:bas} and \ref{section:cad}, we prove that typical CAD algorithms preserve the chordal structure of the associated graph of the polynomial set in question. In Section \ref{section:complexity}, the ``size'' of the full set of projection polynomials obtained in accordance with a perfect elimination ordering is characterized via an $(m,d)$-property, while  Section \ref{section:exp} is devoted to showing the effectiveness of our approach by some examples from applications. Section \ref{sec:conclusion} concludes the paper.

\section{Preliminaries}\label{section:pre}

Denote by $\K[\bar{x}]=\K[x_1,\ldots,x_n]$ the polynomial ring in variables $x_1,\ldots,x_n$ with coefficients in a domain $\K$. For any $x_k$, a nonzero polynomial $f\in\K[\bar{x}]$ can be written in the form $f=\sum_{i=0}^{s}a_{i}x_{k}^{i}$ with $ a_{s}\not\equiv 0,a_{s-1},\ldots,a_0\in\K[x_1,\ldots,x_{k-1},x_{k+1},\ldots,x_{n}]$ and $s\in\Z_{\geq0}$. The number $s$ is called the\emph{ degree} of $f$ with respect to (\emph{w.r.t.}) the variable $x_{k}$ and is denoted by $\degree(f,x_{k})$, while those $a_i$'s are called the {\em coefficients} of $f$ \emph{w.r.t.} $x_k$ and they form a set $\coeff(f,x_k)=\{a_s,a_{s-1},\dots,a_0\}$. In particular, we call $a_{s}$ the {\em leading coefficient} of $f$ \emph{w.r.t.} $x_k$ and denote it by $\lc(f,x_k)$. The {\em content} of $f$ \emph{w.r.t.} $x_k$, denoted by $\cont(f,x_k)$, refers to the GCD of the coefficients $a_{0},a_{1},\ldots a_{s}$. The polynomial $f$ is called {\em primitive} \emph{w.r.t.} $x_k$ if $\cont(f,x_k)=1$ while the polynomial $f/\cont(f,x_k)$ is defined to be \emph{the primitive part} of $f$ \emph{w.r.t.} the variable $x_k$. In addition, we denote by $\cont(F,x_k)$ the set of those non-constant contents of the elements of a polynomial set $F$.





In the following, we recall the definition of the associated graph of a polynomial set and the concept of chordal graph. Let $f$ be a polynomial  in  $\K[\bar{x}]$ and denote by $\var(f)$ the set of variables  which appear in $f$. For a polynomial set $F$, we define $\var(F) =\cup_{f\in F}\;\var(f)$. Then we have the following definition:

\begin{definition}(\cite{mou2019chordal}, Def.1)
Let $F\subset\K[\bar{x}]$ be a finite polynomial set. The undirected graph with the vertex set $\var(F)$ and the edge set $$\{(x_i,x_j)\,|\;\exists f\in F, x_i,x_j\in \var(f)\}$$ is called \emph{the associated graph} of $F$ and is denoted by $\mathscr{G}(F)$.
\end{definition}
\begin{remark}
If $\,\var(F)=\emptyset$, then $\mathscr{G}(F)$ is the empty graph.
\end{remark}

\begin{example}\label{ex:associategraph}
Set $F_{1}=\{ x_{1}x_{3}-1, x_{1}x_{2}-1, x_{2}x_{3}-1 \}$ and $F_{2}=\{ y_{1}^{4}-1, y_{1}^{2}+y_{3}, y_{2}^{2}+y_{3}, y_{3}^{2}+y_{4} \}$. Then the associated graphs $\mathscr{G}(F_1)$ and $\mathscr{G}(F_2)$ of them are shown in Fig. \ref{fig:associategraph1} and Fig. \ref{fig:associategraph2}, respectively.
\vspace{-3mm}
\begin{figure}[htbp]
\centering
\subfigure{
\begin{minipage}[htbp]{0.4\linewidth}
\centering
\begin{tikzpicture}[scale = 0.4,line cap=round,line join=round,>=triangle 45,x=1cm,y=1cm]
\clip(0.9,-0.1) rectangle (7.1,6.1);
\draw [line width=1pt] (2,1) circle (1cm);
\draw [line width=1pt] (2,5) circle (1cm);
\draw [line width=1pt] (6,1) circle (1cm);
\draw [line width=1pt] (2,4)-- (2,2);
\draw [line width=1pt] (3,1)-- (5,1);
\draw [line width=1pt] (2.7355160128001135,4.322492660619368)-- (5.378178567264243,1.7831590552310879);
\begin{scriptsize}
\draw[color=black] (2.0777036305893634,1.0) node[scale=1.5]{$x_1$};
\draw[color=black] (2.0777036305893634,5.0) node[scale=1.5] {$x_2$};
\draw[color=black] (6.067358456306757,1.0) node[scale=1.5] {$x_3$};
\end{scriptsize}
\end{tikzpicture}
\caption{$\;\;\mathscr{G}(F_{1})$}
\label{fig:associategraph1}
\end{minipage}
}
\subfigure{
\begin{minipage}[htbp]{0.4\linewidth}
\centering
\begin{tikzpicture}[scale = 0.4,line cap=round,line join=round,>=triangle 45,x=1cm,y=1cm]
\clip(4.9,-0.1) rectangle (15.1,6.1);
\draw [line width=1pt] (6,5) circle (1cm);
\draw [line width=1pt] (10,5) circle (1cm);
\draw [line width=1pt] (10,1) circle (1cm);
\draw [line width=1pt] (14,5) circle (1cm);
\draw [line width=1pt] (7,5)-- (9,5);
\draw [line width=1pt] (9.999501243850824,1.9999998756211441)-- (10,4);
\draw [line width=1pt] (11,5)-- (13,5);
\begin{scriptsize}
\draw[color=black] (5.967358456306757,5.1) node[scale=1.5] {$y_1$};
\draw[color=black] (9.970356609200462,5.1) node[scale=1.5] {$y_3$};
\draw[color=black] (9.957013282024149,1.1) node[scale=1.5] {$y_2$};
\draw[color=black] (14.1,5.1) node[scale=1.5] {$y_4$};
\end{scriptsize}
\end{tikzpicture}
\caption{$\;\;\mathscr{G}(F_{2})$}
\label{fig:associategraph2}
\end{minipage}
}
\end{figure}
\end{example}

\vspace{-6.5mm}
\begin{definition}(\cite{Blair1993})
Let $G=(V,E)$ be a graph with $V=\{x_{1},\ldots ,x_{n}\}$. An ordering of the vertices $x_{l_{1}}>x_{l_{2}}>\cdots >x_{l_{n}}$ is called \emph{a perfect elimination ordering} if for each vertex $x_{l}$, the induced subgraph with the vertex set
$$X_{l}=\{x_{l}\}\cup \{x_{k}\;|\;x_{k}<x_{l},\;(x_{k},x_{l})\in E\}$$
is a clique. A graph $G$ with a perfect elimination ordering is said to be \emph{chordal}. 
\end{definition}

The graphs $\mathscr{G}(F_1)$ and $\mathscr{G}(F_2)$ in Fig. \ref{fig:associategraph1} and Fig. \ref{fig:associategraph2} are both chordal graphs. In the graph $\mathscr{G}(F_{1})$, any vertex ordering is a perfect elimination ordering. The graph $\mathscr{G}(F_{2})$ is a tree, we know that trees are chordal with perfect elimination orderings obtained by randomly pruning the leafs. 

\begin{definition}(\cite{DBLP:journals/algorithmica/BerryBHP04})
Let $G=(V,E)$ be a graph, \emph{a chordal completion} of $G$ is a chordal graph $\hat{G}=(\hat{V},\hat{E})$ such that $G$ is a subgraph of $\hat{G}$, {\it i.e.} $V \subset \hat{V}$ and $ E\subset\hat{E}$. We call $\hat{G}$ a minimal chordal completion of $G$, if any subgraph of $\hat{G}$ is not a  chordal completion of $G$. 
\end{definition}




For a given graph, there are some known algorithms, {\it e.g.} the {\em elimination game} algorithm in \cite{doi:10.1137/1003021}, to find one of its chordal completions and a corresponding perfect elimination ordering of that chordal completion. 

\begin{definition}
Given a polynomial set $F\subset \K[\bar{x}]$, a graph $G$ is \emph{a chordal structure} of $F$ if $G$ is a chordal completion of $\mathscr{G}(F)$. In particular, if $\mathscr{G}(F)$ is a chordal structure of $F$, then $F$ is said to be \emph{chordal}. For convenience, any chordal graph is regarded as a chordal structure of $\mathscr{G}(F)$ if  $\var(F)=\emptyset$.
\end{definition}
\section{Basic Operations Preserving Chordal Structure}\label{section:bas}
The propositions presented in this section show that some basic polynomial operations, such as discriminant, resultant and subresultant, preserve the chordal structure. Then it is reasonable that any algorithm, containing only these operations, also preserves the chordal structure. Now that the chordal structure is preserved by the algorithm, the variable sparsity embedded in the chordal structure is preserved as well. Later we will see the advantages of preserving such a sparsity pattern during the computation of CAD.

For any $f,g\in\K[\bar{y},x]=\K[y_1,\ldots,y_m,x]$, we denote by $\disc(f,x)$ the discriminant of $f$ \emph{w.r.t.} $x$, and set $\res(f,g,x)$ to be the Sylvester resultant, $S_j$ to be the $j$-th subresultant and $S_{\mu+1},S_{\mu},S_{\mu-1},\cdots,S_{0}$ to be the subresultant chain of $f$ and $g$ \emph{w.r.t.} $x$, respectively. Then we have
\begin{proposition}\label{prop:coeff}
Let $f$ be a polynomial in $\K[\bar{y},x]$ and $G=(V,E)$ a chordal structure of the set $\{f\}$, then $G$ is a chordal structure of the set $\mathfrak{S}=\coeff(f,x)\cup\{\cont(f,x),\disc(f,x)\}$. In particular, $G$ is a chordal structure of the set $\{\lc(f,x)\}$.
\end{proposition}
\begin{proof}
Since $\var(\mathfrak{S})\subset\var(f)$, $\mathscr{G}(\mathfrak{S})$ is a subgraph of $\mathscr{G}(\{f\})$, which is a subgraph of $G$.
\end{proof}
\vspace{-2mm}
\begin{proposition}\label{prop:res}
Let $f,g$ be polynomials in $\K[\bar{y},x]$ and $G=(V,E)$ a chordal structure of the set $\{f,g\}$. If there is a perfect elimination ordering of $G$ such that $x>y_i$ for any $y_i\in\bar{y}$, then $G$ is a chordal structure of the set $\mathfrak{S}=\{\res(f,g,x), S_{\mu +1},S_{\mu},S_{\mu-1},\cdots,S_{0}\}$.
\end{proposition}
\begin{proof}

Since $x>y_i$ for any $y_i\in \bar{y}$, $\var(\mathfrak
{S})\subset \var(\{f,g\})\subset X_x$ with $X_x=\{x\} \cup \{z\in V\;|\;z<x,\, (x,z)\in E\}$ a clique of $G$. So $G$ is a chordal completion of the graph $\mathscr{G}(\mathfrak
{S})$.
\end{proof}

 The {\em finest squarefree basis} of $A\subset\K[\bar{x}]$, denoted by $\mathscr{F}(A)$, is the set of all the irreducible factors of the elements of $A$ (see \cite{DBLP:conf/automata/Collins75}, P.146). Then the following observation is clear:

\begin{proposition}\label{prop:square}
If a finite set $A\subset\K [\bar{x}]$ has a chordal structure $G$, then $G$ is also a chordal structure of the set $\mathscr{F}(A)$ and the set $A\cup\mathscr{F}(A)$.
\end{proposition}
\begin{proof}
Clearly we have  $\var(\mathscr{F}(A))\subset\var(A)$. Since any polynomial in $\mathscr{F}(A)$ is either a polynomial in $A$ or a factor of a polynomial in $A$, $\mathscr{G}(\mathscr{F}(A))$ is a subgraph of $\mathscr{G}(A)$, which is again a subgraph of $G$. 
\end{proof}
\section{CAD Projections With perfect elimination orderings}\label{section:cad}
In this section, we first prove that some well-known CAD algorithms, with the projection operators used in them preserving the chordal structure, also preserve the chordal structure themselves if the projection is done in accordance with a perfect elimination ordering. Based on this, the perfect elimination orderings are recommended while computing the CAD projection. Then, through two examples, we compare the size of the full sets of projection polynomials obtained by the CAD algorithms with respect to perfect elimination orderings and other random variable orderings. The examples show that the perfect elimination orderings can result in much smaller polynomial sets than other naive ones do.


A key concept in the CAD algorithms is the projection operator. In an abstract sense, the projection operator $\proj$ is a mapping which maps an $n$-dimensional polynomial set $P$ containing some variable $x$ to an $(n-1)$-dimensional polynomial set $P'$ such that $x\not\in \var(P')$ and any $P'$-sign-invariant CAD of $\R^{n-1}$ can be induced by a $P$-sign-invariant CAD of $\R^n$. The \emph{projection procedure} of $P$ with variable ordering $x_{n}>\cdots>x_1$ is defined to be the sequence of polynomial sets $\{P_{n},P_{n-1},\ldots,P_1\}$, where $P_n=P$ and each $P_{i}=\proj(P_{i+1},x_{i+1})$ is a set of polynomials in $x_i,\ldots,x_1$ for $i=n-1,\ldots, {1}$. In particular, we call $P_{n-1},\ldots,P_1$ the \emph{projection polynomial sets} \emph{w.r.t.} the ordering $x_{n}>\cdots>x_1$.
 
The projection operator introduced by Collins is defined to be a union of some coefficients, some resultants, some discriminants and some subresultants with respect to a fixed variable $x$. Many improved projection operators are obtained by simplifying Collins' projection operator.

\begin{proposition}\label{prop:proj1}
Set $A\subset\K[\bar{y},x]$ to be a polynomial set. Suppose that a projection operator $\proj(A,x)$ only consists of some coefficients, contents, resultants, discriminants and some subresultants of the polynomials in $A\cup\mathscr{F}(A)$ w.r.t. $x$ and $G$ is a chordal structure of $A$ with a perfect elimination ordering such that $x>y_i$ for any $y_i\in \bar{y}$, then $G$ is also a chordal structure of the polynomial set $\proj(A,x)$. 
\end{proposition}
\begin{proof}
Set $f\in\proj(A,x)$. If $f$ is the discriminant, the content or a coefficient of a polynomial in $A\cup\mathscr{F}(A)$, then $G$ is a chordal structure of $\{f\}$ by Prop. \ref{prop:square} and Prop. \ref{prop:coeff}. If $f$ is a resultant or a subresultant of two polynomials in $A\cup\mathscr{F}(A)$, then by Prop. \ref{prop:square} and Prop. \ref{prop:res}, the same conclusion holds. Hence $G$ is a chordal structure of any $\{f\}\subset\proj(A,x)$, thus also a chordal structure of the set $\proj(A,x)$.
\end{proof}
\begin{proposition}\label{prop:proj2}
Set $A\subset\K[\bar{x}]$ to be a polynomial set. Suppose the graph $G$ is a chordal structure of $A$ and $x_{n}>\cdots>x_1$ is a perfect elimination ordering of $G$. If a projection operator $\proj(S,\cdot)$, with $S$ the polynomial set on which it operates, consists of some coefficients, contents, resultants, discriminants and some subresultants of the polynomials in $S\cup\mathscr{F}(S)$ and $\{P_{n}=A,P_{n-1},\ldots,P_1\}$ is a projection procedure of $A$ obtained via the projection ordering $x_{n}>\cdots>x_1$, then $G$ is a chordal structure of any $P_i$, $i=1,\ldots,n$. 
\end{proposition}
\begin{proof}
The proof is inductive on the index $i$: If $i=n$, then $P_i=A$, the conclusion holds.
Suppose that $G$ is a chordal structure of $P_k$ for some $k\leq n$. Since $P_{k-1}=\proj(P_{k},x_{k})$, $G$ is a chordal structure of $P_{k-1}$ by Prop. \ref{prop:proj1}. The proposition is proved.
\end{proof}


As can be seen from above, the original projection operator and most of the improved projection operators used in the CAD algorithms preserve the chordal structure of a polynomial set. For this reason, when the polynomial system is sparse and with the chordal structure, the perfect elimination orderings of that chordal structure are recommended while computing the CAD projection in order to preserve the chordal structure (as well as the sparsity embedded in the chordal structure) of the system. In this way, the growth of the size of the polynomial sets obtained during the projection can possibly be reduced.

In the following we take care of two popular projection operators and indicate that they preserve the chordal structure. 

\begin{definition}(\cite{DBLP:journals/jsc/McCallum88})
Let $A\subset\Z[\bar{x}]$ be a finite polynomial set containing at least two variables. McCallum's projection operator {$\projmc(\cdot,\cdot)$} is defined as follows:
Let the polynomial set $B$ be the finest squarefree basis of all the primitive parts of the elements of $A$ with positive degrees, then 
\vspace{-1mm}
\begin{align*}
\projmc(A,x_n)=&\;\cont(A,x_{n})\,\cup\bigcup_{f,\,g\,\in\,B,\,f\neq g}\{\res(f,g,x_n)\}\\\vspace{-4mm}
\cup&\bigcup_{f\in B}\big((\coeff(f,x_n)\backslash\{0\})\cup\{\disc(f,x_n)\}\big).
\end{align*}
\end{definition}
\vspace{-1mm}
Since $B\subset\mathscr{F}(A)$, Prop. \ref{prop:proj2} holds for the operator $\projmc$. Indeed, by Prop. \ref{prop:square} and Prop. \ref{prop:proj1}, each set in the projection procedure of any $A\subset\Z[\bar{x}]$ obtained by the operator $\mathscr{F}(\projmc(\cdot,\cdot))$ preserves the chordal structure of $A$ as well.

\begin{example}\label{ex:4-1}
Consider the quantified formula below:
\begin{align*}
\exists x_1.~\;&x_1+x_4>0\wedge x_2+x_4\geq 0\wedge x_{3}^{2}+x_2<0\\
&\quad\wedge x_{3}^{3}+x_1\leq 0\wedge x_5+x_2>0\wedge x_5+x_1+x_2<0.
\end{align*}

The CAD-based methods solving this example rely on constructing a CAD for the set of polynomials \[F=\{
x_{1}+x_{4}, x_{2}+x_{4}, x_{3}^{2}+x_2, x_{3}^{3}+x_1, x_5+x_2, x_5+x_1+x_2\}.\]  The associated graph of $F$ is a chordal graph, shown in Fig. \ref{fig:ex4-1}.

\vspace{-3mm}
\begin{figure}[htbp]
    \centering
    \begin{tikzpicture}[scale=0.4,line cap=round,line join=round,>=triangle 45,x=1cm,y=1cm]
\draw [line width=1pt] (-2,-4) circle (1cm);
\draw [line width=1pt] (4,0) circle (1cm);
\draw [line width=1pt] (0,3) circle (1cm);
\draw [line width=1pt] (2,-4) circle (1cm);
\draw [line width=1pt] (-4,0) circle (1cm);

\draw [line width=1pt] (-3,0)-- (3,0);
\draw [line width=1pt] (3.553,-0.894)-- (2.447,-3.106);
\draw [line width=1pt] (0.275,2.038)-- (1.725,-3.038);
\draw [line width=1pt] (0.8,2.4)-- (3.2,0.6);
\draw [line width=1pt] (-3.168,-0.555)-- (1.168,-3.445);
\draw [line width=1pt] (-1.168,-3.446)-- (3.168,-0.555);
\draw [line width=1pt] (-1,-4)-- (1,-4);

\begin{scriptsize}
\draw[color=black] (-2,-4) node[scale=1.4] {$x_3$};
\draw[color=black] (4,0) node[scale=1.4] {$x_2$};
\draw[color=black] (0,3) node[scale=1.4] {$x_4$};
\draw[color=black] (2,-4) node[scale=1.4] {$x_1$};
\draw[color=black] (-4,0) node[scale=1.4] {$x_5$};
\end{scriptsize}
\end{tikzpicture}
\vspace{-1mm}
    \caption{\;$\mathscr{G}(F)$}
    \label{fig:ex4-1}
\end{figure}
\vspace{-3mm}

We successively use the operator $\mathscr{F}(\projmc(\cdot,\cdot))$ in accordance with the perfect elimination ordering $x_4>x_5>x_3>x_2>x_1$ and obtain a projection procedure:
\begin{align*}
{F}'=&\mathscr{F}(\projmc(F,x_4))\\
=&\{x_{1},x_{2},x_{3}^{2}+x_2, x_{3}^{3}+x_1, x_5+x_2, x_5+x_1+x_2,x_1-x_2\},\\
{F}''=&\mathscr{F}(\projmc({F}',x_5))\\
=&\{x_{3}^{2}+x_2,x_{3}^{3}+x_1,x_1-x_2,x_1+x_2,x_1,x_2\},\\
{F}'''=&\mathscr{F}(\projmc({F}'',x_3))\\ 
=&\{x_1-x_2,x_1+x_2,x_1,x_2,x_{2}^{3}+x_{1}^2\},\\
{F}''''=&\mathscr{F}(\projmc({F}''',x_2))\\ 
=&\{x_{1},x_{1}+1,x_1-1\}.
\end{align*}

When a random variable ordering, say, $x_1>x_2>x_3>x_4>x_5$ is used instead, we finally obtain: 
\begin{align*}
F''''=\{&x_{5},x_{5}-1,x_5+1,x_{5}-2,x_5-8,4x_{5}-1,8x_5-1,27x_{5}-4,\\
&3x_5-1,x_{5}^{2}+x_{5}-1,x_5^2+x_5+1,x_5^2-x_5+1,x_{5}^{2}-3x_{5}+9,\\
&x_{5}^{3}-x_{5}^{2}+2x_{5}-1,x_5^5-3x_5^4+3x_5^3+5x_5^2+2x_5-1,\\
&x_5^4+x_5^3+x_5^2+x_5+1,x_{5}^{5}+4x_{5}^{3}-x_{5}^{2}+2x_{5}-1,\\
&-x_5^3+4x_5^2-3x_5+1,x_5^5-3x_5^4-6x_5^3-19x_5^2+9x_5-1\}.
\end{align*}

The above process presents the results of the computation with respect to two specific variable orderings, while Table \ref{tab:diff-order} shows the experimental results with respect to two classes of variable orderings with the first two maximal variables being $x_4>x_5>\cdots$ and $x_1>x_2>\cdots$, respectively.
Therein, the ``PEO'' column shows whether the variable orderings with the corresponding maximal variables in the ``Max-Var'' column are perfect elimination orderings or not. By ``\#proj'' we denote the sum of the size of those projection polynomial sets obtained in the CAD projection \emph{w.r.t.} a variable ordering. Finally, the ``avg'' column shows the average of those \#proj values with respect to the six orderings corresponding to each pair of maximal variables. 

For this example, it is true that the CAD projection tends to generate fewer polynomials when a perfect elimination ordering is used instead of a random one. 

\vspace{-2mm}
\begin{table}[htbp]
    \centering
    \begin{tabular}{c|c|c|c|c|c}
\toprule
Max-Var & PEO & \multicolumn{3}{c}{\#proj} & \multicolumn{1}{|c}{avg}\\
\hline
\multirow{2}{*}{$\;\;\;x_4>x_5>\cdots$}& \multirow{2}{*}{yes}& 21 & 20 & 21 &\multirow{2}{*}{20.7}\\ 
\cline{3-5}
  & &  \multicolumn{1}{c|}{21} & 20 & 21 &\\
\hline
\multirow{2}{*}{$\;\;\;x_1>x_2>\cdots$}&
\multirow{2}{*}{no}& 53 & 74 & 34 &\multirow{2}{*}{48.8}\\
\cline{3-5}
 &   & \multicolumn{1}{c|}{33} & 52 & 47   &
\\
\bottomrule
\end{tabular}
\vspace{1mm}
    \caption{Different variable orderings in
    Example \ref{ex:4-1}.}
    \label{tab:diff-order}
\end{table}

\end{example}

\vspace{-11mm}
\begin{definition}(\cite{DBLP:journals/jsc/Brown01a}) Let $A\subset\Z[\bar{x}]$ be a squarefree basis \cite{DBLP:conf/automata/Collins75} with at least two variables. 
Brown's projection operator $\projbrown(\cdot,\cdot)$ is defined as follows:
Set $B=\{f\in A\;|\;x_n\in\var(f)\}$, then
\begin{align*}
\projbrown(A,x_n)=\;\cont(A,x_{n})\,&\cup\\
\bigcup_{f,\,g\,\in\,B,\,f\neq g}\{\res(f,g,x_n)\}&\cup\bigcup_{f\in B}\{\lc(f,x_n),\disc(f,x_n)\}.\\
\end{align*}
\end{definition}
\vspace{-5mm}
Similarly, Prop. \ref{prop:proj2} holds for Brown's operator $\projbrown$. Also, by Prop. \ref{prop:square} and Prop. \ref{prop:proj1}, each set in the projection procedure of any $A\subset\Z[\bar{x}]$ obtained by the operator $\mathscr{F}(\projbrown(\cdot,\cdot))$ preserves the chordal structure of $A$ as well.

\begin{example}\label{ex:4-2}
Let $$F=\{x_{1}+x_{2}+2,x_{2}x_{3}+2x_{3}+x_{1},x_{3}x_{4}+x_{2}x_{4}+x_{3}-1,x_{4}+x_{2}\}$$ be a squarefree basis.
The associated graph $\mathscr{G}(F)$ of $F$ is a chordal graph with $x_{1}>x_{2}>x_{3}>x_{4}$ a perfect elimination ordering. As in Example \ref{ex:4-1}, by  successively applying the operator $\mathscr{F}(\projbrown(\cdot,\cdot))$ in that order, we can finally obtain the full set of projection polynomials: $\{x_{4},x_{4}+1,x_{4}-1,x_{4}-2\}$.

If a random projection ordering, say, $x_{2}>x_{3}>x_{4}>x_{1}$, which is not a perfect elimination ordering, is used in the CAD projection, then we will finally obtain the following polynomial set:
\begin{align*}
\{&x_1, 5+4x_1, 9+8x_1, 25+24x_1, x_1+1, x_1+2, x_1+3, 3x_1+4,\\
&2x_1+5, 32x_1^2+56x_1+25, 20x_1^2+44x_1+25, 13x_1^2+34x_1+25,\\
&4x_1^3+24x_1^2+44x_1+25,4x_1^3-71x_1^2-172x_1-100,\\
& x_1^2+4x_1+5,x_1^4+7x_1^3+21x_1^2+34x_1+25\}.
\end{align*}

\vspace{-3mm}
\begin{table}[htbp]
    \centering
    \begin{tabular}{c|c|c|c|c|c}
\toprule
Max-Var & PEO & \multicolumn{3}{c}{\#proj} & \multicolumn{1}{|c}{avg} \\
\hline
\multirow{2}{*}{$\;\;\;x_1>\cdots$} & \multirow{2}{*}{yes}  & 13 & 13 & 13 & \multirow{2}{*}{12.7}\\ 
\cline{3-5}
&  & 13 & 12 & 12 &\\ 
\hline
\multirow{2}{*}{$\;\;\;x_2>\cdots$} & \multirow{2}{*}{no}  & 26 & 25 & 35 & \multirow{2}{*}{33.3}\\
\cline{3-5}
    &    & 37 &36 & 41 & \\
\hline
\multirow{2}{*}{$\;\;\;x_3>\cdots$} & \multirow{2}{*}{no} & 14  & 14  & 22 & \multirow{2}{*}{17.2}\\
\cline{3-5}
   &     & 21 & 15   & 17  & \\  
\hline
\multirow{2}{*}{$\;\;\;x_4>\cdots$} & \multirow{2}{*}{yes}   & 12  & 12 & 19 & \multirow{2}{*}{16.2}\\
\cline{3-5}
   &  & 24  & 14    & 16 & \\
\bottomrule
\end{tabular}
\vspace{1mm}
    \caption{Different variable orderings in Example \ref{ex:4-2}.}
    \label{tab:diff-order2}
\end{table}
\vspace{-7mm}

Moreover, the sum of the size of the projection polynomial sets obtained via each of the 24 variable orderings is shown in Table \ref{tab:diff-order2}. Again we see that, when the perfect elimination orderings are used, the CAD projection may yield smaller polynomial sets than it does when other kind of variable orderings are used instead.

\end{example}

\section{Complexity Analysis: the ``Size'' of the Projection Polynomials }\label{section:complexity}
In this section, we discuss the complexity of the CAD computation via a perfect elimination ordering by giving a new estimate of the 
``size'' of the full set of projection polynomials obtained via that ordering. The results in this section are obtained by the method of Bradford et al. \cite{DBLP:journals/jsc/BradfordDEMW16}, based on which they estimate the growth of the ``size'' of the projection polynomials.


For a set of polynomials $A\subset\Z[\bar{x}]$, \emph{the combined degree} \cite{McCallum85} of $A$ refers to the number $\max_{x\in\var(A)}\deg(\prod A,x)$, where $\prod A$ is the product of all the polynomials in $A$. And we say the set $A$ has \emph{the} $(m,d)$\emph{-property} \cite{DBLP:journals/jsc/BradfordDEMW16} if it can be partitioned into at most $m$ pairwise disjoint subsets, each with combined degree at most $d$.
This type of properties are useful while tracking the size of the projection polynomials.


The following lemma estimates the growth of the ``size'' of the projection polynomial sets while using the operator $\projmc$.

\begin{lemma}(\cite{DBLP:journals/jsc/BradfordDEMW16}, Lemma 11, Corollary 12)\label{lem:m2}
Suppose $A\subset\R[\bar{x}]$ is a set of polynomials with the $(m,d)$-property and $x\in\var(A)$, then the set \emph{$\projmc(A,x)$} has the $(M, 2d^2)$-property with
$M=\left\lfloor \frac{(m+1)^2}{2}\right\rfloor$. If $m>1$, then the set \emph{$\projmc(A,x)$} has the $(m^2, 2d^2)$-property.
\end{lemma}

In Subsection \ref{section:com_general} we recall the general complexity analysis for CAD while in Subsection \ref{section:com_per} we provide the complexity analysis with respect to the case where a perfect elimination ordering is used during the CAD projection.

\vspace{-1.5mm}
\subsection{Complexity Analysis for CAD with General Variable Ordering }\label{section:com_general}
 
In Table \ref{tab:cad_general_growth} (\cite{DBLP:journals/jsc/BradfordDEMW16} Table 1), Lemma \ref{lem:m2} is applied recursively to estimate the growth of the pair $(m,d)$ in the $(m,d)$-property during the CAD projection rendered by the operator $\projmc$.  The  ``Number'' column shows the upper bound of the number of the subsets in a possible partition and the  ``Degree'' column shows the upper bound of the  corresponding combined degree of those subsets, and $M=\left\lfloor \frac{(m+1)^2}{2}\right\rfloor$.
 
\vspace{-3mm}
\begin{table}[htbp]
    \centering
    \caption{Growth of the ``size'' of the polynomial sets in CAD projection.}
    \vspace{-4mm}
    \begin{tabular}{c c c}
       
        Variables & Number        & Degree            \\
        \hline
        \hline
        $n$       & $m$           & $d$               \\
        \hline
        $n-1$     & $M$           & $2d^2$            \\
        $n-2$     & $M^2$         & $8d^4$            \\
        $\vdots$  & $\vdots$      & $\vdots$          \\
        $n-r$       & $M^{2^{r-1}}$ & $2^{2^r-1}d^{2^r}$\\
        $\vdots$  & $\vdots$      & $\vdots$          \\
        $1$       & $M^{2^{n-2}}$ & $2^{2^{n-1}-1}d^{2^{n-1}}$
    \end{tabular}
    \label{tab:cad_general_growth}
\end{table}
\vspace{-4mm}

The number of the real roots of the product of the polynomials in a univariate polynomial set with the $(m,d)$-property is at most $md$. And the number of the corresponding cells in $\R^1$ is twice the number of the real roots plus one. Therefore, the total number of the cells in the CAD of $\R^n$ is at most the product of those $2{m'}{d'}+1$, with ${m'}$ and ${d'}$ taking the corresponding values in the ``Number'' and the ``Degree'' columns in Table \ref{tab:cad_general_growth}, \emph{i.e.}
$$ (2md+1)\prod_{r=1}^{n-1}\left[2(2^{2^r-1}M^{2^{r-1}}d^{2^r})+1\right].$$
So the following estimate is obtained:

\vspace{-1.5mm}
\begin{theorem}\label{thm:general-cells}(\cite{DBLP:journals/jsc/BradfordDEMW16})
The number of the CAD cells in $\R^n$ obtained by the operator \emph{$\projmc$}  is
$O((2d)^{2^n-1}M^{2^{n-1}-1}m)$ with $M=\left\lfloor \frac{(m+1)^2}{2}\right\rfloor$.
\end{theorem}
\vspace{-3.5mm}
\subsection{Complexity Analysis for CAD with Perfect Elimination Ordering}\label{section:com_per}


First, we introduce the concept of elimination tree to further exploit the relationship between the variables during the perfect elimination process. For convenience, we assume  in this section that all the associated graphs  of the polynomial sets are connected.

\begin{definition}(\cite{cifuentes2017chordal} Def. 2.5)\label{def:elim_tree}
Let $G$ be a chordal graph with a perfect elimination ordering $x_1<x_2<\ldots<x_n$. \emph{The elimination tree} of $G$ is the following directed spanning tree: for each $l\neq1$ there is an arc from $x_l$ toward the
largest $x_p$ that is adjacent to $x_l$ and $x_p < x_l$. We will say that $x_p$ is the parent of $x_l$ and $x_l$ is a child of $x_p$ (we denote by $\child(x_p)$ the set of all the children of $x_p$). Note that the elimination tree is rooted at $x_1$.
\end{definition}

\vspace{-2mm}
\begin{lemma}\label{lemma:path}
Let $T$ be the elimination tree of a chordal graph $G$. If $x_s<x_t$ and $(x_s,x_t)$ is an edge of $G$, then there is a path from $x_t$ to $x_s$ in $T$.
\end{lemma}
\vspace{-2mm}
\begin{proof}
If $x_s$ is the parent of $x_t$, there is nothing to prove. Otherwise, let $x_{t'}$ be the parent of $x_t$. By Def. \ref{def:elim_tree}, $x_s<x_{t'}$. Because $(x_s,x_t)$ and $(x_{t'},x_t)$ are edges of $G$ which is a chordal graph, $(x_s,x_{t'})$ is also an edge of $G$. We are done by repeating the same discussion.
\end{proof}

We will see that the path from a node to the root in the elimination tree is exactly the projection path of the polynomials with the variable corresponding to this node being the maximal variable. It is in this way the elimination tree describes the relationships between those variables, and this enables us to give a new complexity analysis.

Let $A\subset\Z[\bar{x}]$ be a polynomial set, $G$ a chordal structure of $A$ with a perfect elimination ordering  $x_1<\ldots<x_n$ and $T$ the elimination tree of $G$ with respect to that ordering. 
We denote by $P_i$ the projection polynomial sets of the operator $\projmc$, {\it i.e.}
$$P_n=A\,\text{\;and }\, P_i=\text{$\projmc$}(P_{i+1},x_{i+1})\;~ \text{for}~ i=n-1,\ldots,1.$$
We define $T_p$ to be the projection polynomial sets of the operator $\projmc$ along the elimination tree as follows:
\begin{definition}\label{def:tx}
For each $1\leq p\leq n$ we define 
$$A_{p}=\{f\in A\;|\,\max(\var(f))=x_p\}$$
and 
\begin{equation}\label{tree_elimination}
T_{p}=A_{p}\cup \bigcup_{x_l\,\in\,\child(x_p)} \text{$\projmc$}(T_{l},x_l).
\end{equation}
Note that when $x_p$ is a leaf of $T$, $T_{p}=A_{p}$.
\end{definition}

\begin{proposition}\label{prop:SetEqu}
For $i=1,\ldots,n$, we have
$$\{f\in T_i\;|\;x_i\in\var(f)\} =\{f\in P_i\;|\;x_i\in\var(f)\}.$$
Furthermore, we have
$$\{f\in \bigcup_i T_i\;|\;\var(f)\neq \emptyset\}=\{f\in \bigcup_i P_i\;|\;\var(f)\neq \emptyset\}.$$
\end{proposition}
\vspace{-2mm}
\begin{proof}
It is clear that  $\{f\in T_i\;|\;x_i\in\var(f)\}\subset P_i$. So we only  need to prove that for any $1\leq p\leq n$,
\begin{equation}\label{containcccc}
\{f\in P_p\;|\;x_p\in\var(f)\}\subset T_p.
\end{equation}

When $p=n$, the claim holds. Assume that the containment (\ref{containcccc}) holds for $p=n,n-1,\ldots,k+1$. In the following we prove that it also holds for $p=k$.
Set $g\in P_k$ to be a polynomial such that $x_k\in\var(g)$ and set $l=\max(\{1\leq i\leq n\;|\;g\in P_i\})\geq k$. If $l=n$, then $g\in A_k\subset T_k$.  Otherwise $l<n$, and  $g\notin P_{l+1}$ but $g\in P_l$. Suppose that the polynomial $g$ is constructed from a polynomial $h\in P_{l+1}$ (or from two polynomials $h_1$, $h_2\in P_{l+1}$), then we have $x_{l+1}\in\var(h)$ (\emph{resp.}, $x_{l+1}\in\var(h_1)\cap\var(h_2)$). Since $x_k\in\var(g)$, $x_k\in\var(h)$ (\emph{resp.}, $x_k\in\var(h_1)\cup\var(h_2)$). So there is a polynomial ($h,h_1$ or $h_2$) that contains both the variables $x_k$ and $x_{l+1}$.
This means that  $(x_k,x_{l+1})$ is an edge of $G$. By Lemma \ref{lemma:path}, there is a path from $x_{l+1}$ to $x_k$ in $T$. On the other hand, $l+1>k$. So by the inductive assumption, (\ref{containcccc}) holds for $p=l+1$ and we have $h\in T_{l+1}$ (\emph{resp.}, $h_1,h_2\in T_{l+1}$). Thus $g\in \projmc(T_{l+1},x_{l+1})\subset T_q$ with $q$ such that $x_q$ is the parent of $x_{l+1}$ in $T$. Since $\max(\var(g))=x_k$, 
$g$ remains unchanged in the projection process along the sub-path from $x_q$ to $x_k$ in $T$ (which is defined by (\ref{tree_elimination})). Therefore, $g\in T_k$.



Further, for every $g\in \bigcup_iP_i$ so that $\var(g)\neq \emptyset$. Set $k=\min(\{1\leq i \leq n\,|\,g\in P_i\})$, then  $x_k\in \var(g)$ and $g\in T_k$.   Therefore $\{f\in \bigcup_i P_i\;|\;\var(f)\neq \emptyset\}\subset \{f\in \bigcup_i T_i\;|\;\var(f)\neq \emptyset\}$. The reverse containment can be proved similarly.
\end{proof}
\vspace{-2mm}
The above proposition shows that the difference between $\cup T_i$ and $\cup P_i$ merely consists of some constants, which make no differences in the CAD computation. So we can estimate the growth of the ``size'' of the projection polynomials by estimating the ``size'' of $T_i$. 

\vspace{-1mm}
\begin{lemma}\label{lem:T_m2}
For a variable $x_p$, if the set $A_p$ has the  $(m,d)$-property and the set $T_l$ also has this property for each $x_l\in\child(x_p)$, then $T_p$ has the $\big((|\child(x_p)|+1)M,2d^2\big)$-property with $M=\left\lfloor \frac{(m+1)^2}{2}\right\rfloor$.  When $m>1$, the set $T_p$ has the $\big((|\child(x_p)|+1)m^2,2d^2\big)$-property.
\end{lemma}
\vspace{-4mm}
\begin{proof}
By Def. \ref{def:tx}, $T_p=A_p\cup \bigcup\limits_{x_l\,\in\,\child(x_p)} \projmc(T_l,x_l)$.
For any $x_l\in \child(x_p)$, the set $\projmc(T_l,x_l)$ has the $(M,2d^2)$-property with $M=\left\lfloor \frac{(m+1)^2}{2}\right\rfloor$ by Lemma \ref{lem:m2}. And $A_p$ has the  $(m,d)$-property. Since $m\leq M$ and $d<2d^2$, $A_p$ also has the  $(M,2d^2)$-property.
\end{proof}
\vspace{-5mm}
\begin{table}[htbp]
    \centering
    \caption{Growth of the ``size'' of the polynomial sets in CAD projection with perfect elimination ordering. }\label{tab:cad_pre_growth}
    \vspace{-3mm}
    \begin{tabular}{c c c}
        Height & Number &  Degree  \\
        \hline
        \hline
        $0$      & $m$                  & $d$ \\
        \hline
        $1$      & $(w+1)M$          & $2d^2$\\
        $2$      & $(w+1)^3M^2$   & $8d^4$\\
        $3$      & $(w+1)^7M^4$ & $128d^8$\\
        $\vdots$ & $\vdots$             & $\vdots$\\
        $r$      & $(w+1)^{2^r-1}M^{2^{r-1}}$ & $2^{2^r-1}d^{2^r}$\\
        $\vdots$ & $\vdots$             & $\vdots$\\
        $h$      & $(w+1)^{2^h-1}M^{2^{h-1}}$ & $2^{2^h-1}d^{2^h}$\\
        
    \end{tabular}
    
\end{table}
\vspace{-3mm}

 As in Section \ref{section:com_general}, we apply Lemma \ref{lem:T_m2} recursively to estimate the growth of the ``size'' of the sets $T_i$. The results are shown in Table \ref{tab:cad_pre_growth} with $M=\left\lfloor \frac{(m+1)^2}{2}\right\rfloor$ and $w=\max_{1\leq l \leq n}(|\child(x_l)|)$ therein. This is how one reads the table: a triple $(h',m',d')$ shown in a row means that for any node $x_i$ with height $h'$ in $T$, the set $T_i$ has the $(m',d')$-property. The main difference of this table from Table \ref{tab:cad_general_growth} is that we replace the number of variables therein by the height of the nodes in the perfect elimination tree. Rigorously, we have

\begin{theorem}\label{ttttthm2}
Suppose the set $A_l$ has the $(m,d)$-property for any $1\leq l\leq n$. Then for every internal node  $x_i$ with height $h_i$, the set $T_i$ has the $$\big((w+1)^{2^{h_i}-1}M^{2^{h_i-1}},2^{2^{h_i}-1}d^{2^{h_i}}\big)\text{-property}$$ with $w=\max_{1\leq l \leq n}(|\child(x_l)|)$ and $M=\left\lfloor \frac{(m+1)^2}{2}\right\rfloor$. Note that for a leaf node $x_l$, the set $T_l=A_l$ has the $(m,d)$-property by assumption.
\end{theorem}

The following theorem gives the upper bound of the number of the CAD cells according to Theorem \ref{ttttthm2} and Prop. \ref{prop:SetEqu}.

\begin{theorem}\label{thm:PEO-cells}
If the set $A_l$ has the $(m,d)$-property for every $x_l$, then the number of CAD cells in ${\mathbb R}^n$ is at most
$\prod_{i=1}^{n}(2K_i+1),$
where 
$$
K_i= 
\left\{\begin{array}{cl} 
md ,& \text{if $x_i$ is a leaf node}\\
\big(2(w+1)\big)^{2^{h_i}-1}M^{2^{h_i-1}}d^{2^{h_i}},& \text{ otherwise}
\end{array}\right.
$$
with $w=\max_{1\leq l \leq n}(|\child(x_l)|)$, $M=\left\lfloor \frac{(m+1)^2}{2}\right\rfloor$ and $h_i$ the height of the node $x_i$ in the tree $T$. 
\end{theorem}

Comparing Table \ref{tab:cad_pre_growth} with Table \ref{tab:cad_general_growth}, we can see some significant differences. The first one is that the $(m,d)$-property is for $A$ in Table \ref{tab:cad_general_growth}, but it is for the subsets $A_l\subset A$ in Table \ref{tab:cad_pre_growth}. This means Table \ref{tab:cad_pre_growth} (hence Theorem \ref{thm:PEO-cells}) may accept a smaller pair $(m,d)$ than Table \ref{tab:cad_general_growth} (\emph{resp}. Theorem \ref{thm:general-cells}) does. The second difference is the quantity in the double-exponent: one is the number of variables while the other is the height of the node. In general, the difference between these two is not big. But in the case of sparsity, the latter estimate will be much smaller than the former one. This is consistent with our guess at the very beginning. Also, there is a double exponent with $w+1\geq2$ as the base in the new estimate, but note that $w$ describes the number of branches and the height is lower if there are more branches, and vice versa. So this won't be too bad for the new estimate.
\section{Experiments}\label{section:exp}

In this section, we compare the performance of the CAD computation via perfect elimination orderings and via other random variable orderings by testing some polynomial sets with variable sparsity.
As explained before, computing CAD projection via a perfect elimination ordering preserves the chordal structure (as well as the pattern of sparsity embedded in it). This is how we exploit the variable sparsity in the CAD computation. And because of this, the CAD computation via a perfect elimination ordering usually performs better on the examples with variable sparsity.

All tests were conducted on 6-Core Intel Core i7-8750H@2.20GHz with 32GB of memory and ARCH LINUX SYSTEM. 

\subsection{The Impact of the Tree Height}

According to the result in Section \ref{section:complexity}, the new complexity analysis for the CAD computation via a perfect elimination ordering has two main differences from the original one for the CAD computation via a general variable ordering: one is the new $(m,d)$-property and the other is the height of the elimination tree introduced in the formulae. Among them, the impact of the tree height on the new complexity  is particularly noteworthy, and the height of the elimination trees with respect to different perfect elimination orderings  of the same chordal structure can be  different. We recommend not only the perfect elimination orderings as the variable orderings used in the CAD projection, but also the ones among them which result in perfect elimination trees of the minimal height. The following example shows how this can be used to reduce the computation.
\vspace{-1mm}

\begin{example}\label{ex:F1}

The following polynomial set forms a lattice reachablility problem which is described in \cite{Diaconis1998}:
\vspace{-0.5mm}
\[
    F_n=\{x_kx_{k+3}-x_{k+1}x_{k+2}\;|\;k=1,2,\ldots,n-3\}, 4\leq n\in \Z.
    \vspace{-0.5mm}
\]
The associated graph $\mathscr{G}(F_n)$ of the set $F_n$ is shown in Fig. \ref{fig:F1} and it is obviously already a chordal structure of $F_n$. 
\vspace{-1mm}
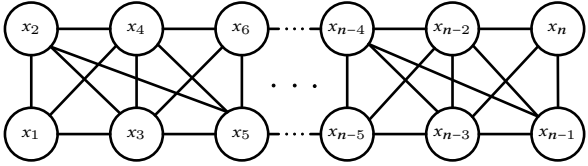
\begin{figure}[htbp]
\centering
\begin{tikzpicture}[scale = 0.35,line cap=round,line join=round,>=triangle 45,x=1cm,y=1cm]
\clip(0.9,-0.1) rectangle (23.1,6.1);
\draw [line width=1pt] (2,1) circle (1cm);
\draw [line width=1pt] (2,5) circle (1cm);
\draw [line width=1pt] (6,1) circle (1cm);
\draw [line width=1pt] (6,5) circle (1cm);
\draw [line width=1pt] (10,5) circle (1cm);
\draw [line width=1pt] (10,1) circle (1cm);
\draw [line width=1pt] (11,1)-- (11.4,1);
\draw [fill=black] (11.7,1) circle (1pt);
\draw [fill=black] (12,1) circle (1pt);
\draw [fill=black] (12.3,1) circle (1pt);
\draw [line width=1pt] (12.6,1)-- (13,1);
\draw [line width=1pt] (2,4)-- (2,2);
\draw [line width=1pt] (3,1)-- (5,1);
\draw [line width=1pt] (3,5)-- (5,5);
\draw [line width=1pt] (6,2)-- (6,4);
\draw [line width=1pt] (2.7031777104103822,1.711014140212424)-- (5.230201993828059,4.361712423986097);
\draw [line width=1pt] (2.7355160128001135,4.322492660619368)-- (5.378178567264243,1.7831590552310879);
\draw [line width=1pt] (7,5)-- (9,5);
\draw [line width=1pt] (7,1)-- (9,1);
\draw [line width=1pt] (9.999501243850824,1.9999998756211441)-- (10,4);
\draw [line width=1pt] (6.804385809896945,4.405892712688663)-- (9.454529808398158,1.8381302226228633);
\draw [line width=1pt] (6.697391092358913,1.7166907731361147)-- (9.270509847045185,4.316008686647293);
\draw [line width=1pt] (12.6,5)-- (13,5);
\draw [fill=black] (11.7,5) circle (1pt);
\draw [fill=black] (12,5) circle (1pt);
\draw [fill=black] (12.3,5) circle (1pt);
\draw [line width=1pt] (11,5)-- (11.4,5);
\draw [line width=1pt] (14,1) circle (1cm);
\draw [line width=1pt] (14,5) circle (1cm);
\draw [line width=1pt] (18,1) circle (1cm);
\draw [line width=1pt] (18,5) circle (1cm);
\draw [line width=1pt] (22,5) circle (1cm);
\draw [line width=1pt] (22,1) circle (1cm);
\draw [line width=1pt] (14,4)-- (14,2);
\draw [line width=1pt] (15,1)-- (17,1);
\draw [line width=1pt] (15,5)-- (17,5);
\draw [line width=1pt] (18,2)-- (18,4);
\draw [line width=1pt] (14.838901252543753,1.544283647081669)-- (17.3214044611064,4.26548785265749);
\draw [line width=1pt] (14.816117806013656,4.422114434591538)-- (17.451705879051282,1.8362855714007464);
\draw [line width=1pt] (19,5)-- (21,5);
\draw [line width=1pt] (19,1)-- (21,1);
\draw [line width=1pt] (22.008753456384277,1.9999616877667505)-- (22,4);
\draw [line width=1pt] (18.775338262878122,4.36845381948971)-- (21.254521968397672,1.666530197664231);
\draw [line width=1pt] (18.75567912923812,1.6549420231088512)-- (21.2515492585613,4.336809614333967);
\draw [line width=1pt] (2.7355160128001135,4.322492660619368)-- (9.454529808398158,1.8381302226228633);
\draw [line width=1pt] (14.816117806013656,4.422114434591538)-- (21.254521968397672,1.666530197664231);
\begin{scriptsize}
\draw[color=black] (2,1) node[scale =1]{$x_1$};
\draw[color=black] (2,5) node[scale =1] {$x_2$};
\draw[color=black] (6,1) node[scale =1] {$x_3$};
\draw[color=black] (6,5) node[scale =1] {$x_4$};
\draw[color=black] (10,5) node[scale =1] {$x_6$};
\draw[color=black] (10,1) node[scale =1] {$x_5$};
\draw[color=black] (14,1) node[scale =1] {$x_{n-5}$};
\draw[color=black] (14,5) node[scale =1] {$x_{n-4}$};
\draw[color=black] (18,1) node[scale =1] {$x_{n-3}$};
\draw[color=black] (18,5) node[scale =1] {$x_{n-2}$};
\draw[color=black] (22,5) node[scale =1] {$x_n$};
\draw[color=black] (22,1) node[scale =1] {$x_{n-1}$};
\draw [fill=black] (11.20467718494546,2.845699927087352) circle (1.5pt);
\draw [fill=black] (12.00658910345555,2.845699927087352) circle (1.5pt);
\draw [fill=black] (12.765154431775905,2.8240266319924845) circle (1.5pt);
\end{scriptsize}
\end{tikzpicture}
\vspace{-2mm}
    \caption{ Associated graph $\mathscr{G}(F_n)$ of  $F_n$ in Example \ref{ex:F1}.}
    \label{fig:F1}
    \vspace{-2mm}
\end{figure}

We consider two different perfect elimination orderings 
\begin{align*}
    l_n^1&:x_1>x_2>x_3>\ldots>x_{n-1}>x_{n},\text{ and}\\
    l_n^2&:x_1>\dots>x_{\left\lceil\frac{n-3}{2}\right\rceil}>x_{n}>\ldots>x_{\left\lceil \frac{n-1}{2}\right\rceil}.
\end{align*}
The figures below (Fig. \ref{fig:ln1} and Fig. \ref{fig:ln2}) show the elimination trees with respect to the ordering $l_n^1$ and $l_n^2$, respectively.
\vspace{-5mm}
\begin{figure}[htbp]
\centering
\subfigure{
\begin{minipage}[htbp]{0.35\linewidth}
\centering
\rotatebox{-90}{
\begin{tikzpicture}[scale = 0.2,line cap=round,line join=round,>=triangle 45,x=1cm,y=1cm]
\clip(-8.2,-4) rectangle (12.2,4);
\draw [line width=1pt] (-4,0) circle (2cm);
\draw [line width=1pt] (10,0) circle (2cm);
\draw [-latex,line width=0.8pt] (-2,0) -- (1.36,0);
\draw [-latex,line width=0.8pt] (4.95,0) -- (8,0);

\begin{scriptsize}
\draw[color=black] (-4,0) node[scale = 1] {\rotatebox{90}{$x_1$}};
\draw[color=black] (10,0) node[scale = 1] {\rotatebox{90}{$x_n$}};
\draw [fill=black] (2.2575,0) circle (2.5pt);
\draw [fill=black] (3.155,0) circle (2.5pt);
\draw [fill=black] (4.0525,0) circle (2.5pt);
\end{scriptsize}
\end{tikzpicture}}
\caption{The elimination tree with respect to  $l_n^1$.}
\label{fig:ln1}
\end{minipage}}
\subfigure{
\begin{minipage}[htbp]{0.6\linewidth}
\centering
\rotatebox{0}{
\begin{tikzpicture}[scale = 0.4,line cap=round,line join=round,>=triangle 45,x=1cm,y=1cm]
\clip(-2.6,-4.6) rectangle (10.8,5.4);
\draw [line width=1pt] (0,0.5) circle (1cm);
\draw [line width=1pt] (3.5,0.5) circle (1cm);
\draw [line width=1pt] (7,0.5) circle (1cm);
\draw [line width=1pt] (0,4) circle (1cm);
\draw [line width=1pt] (0,-3) circle (1cm);
\draw [line width=1pt] (7,4) circle (1cm);
\draw [line width=1pt] (7,-3) circle (1cm);
\draw [-latex,line width=0.7pt] (0,3) -- (0,1.5);
\draw [-latex,line width=0.7pt] (0,-2) -- (0,-0.5);
\draw [-latex,line width=0.7pt] (1,0.5) -- (2.5,0.5);
\draw [-latex,line width=0.7pt] (4.5,0.5) -- (6,0.5);
\draw [-latex,line width=0.7pt] (6,4) -- (4.5,4);
\draw [-latex,line width=0.7pt] (2.5,4) -- (1,4);
\draw [-latex,line width=0.7pt] (6,-3) -- (4.5,-3);
\draw [-latex,line width=0.7pt] (2.5,-3) -- (1,-3);
\draw[color=black] (0,0.5) node[scale =0.7]
{\rotatebox{0}{$x_{\left\lceil\frac{n+3}{2}\right\rceil}$}};
\draw[color=black] (3.5,0.5) node[scale =0.7]
{\rotatebox{0}{$x_{\left\lceil\frac{n+1}{2}\right\rceil}$}};
\draw[color=black] (7,0.5) node[scale =0.7]
{\rotatebox{0}{$x_{\left\lceil\frac{n-1}{2}\right\rceil}$}};
\draw[color=black] (0,4) node[scale =0.7]
{\rotatebox{0}{$x_{\left\lceil\frac{n-3}{2}\right\rceil}$}};
\draw[color=black] (0,-3) node[scale =0.7]
{\rotatebox{0}{$x_{\left\lceil\frac{n+5}{2}\right\rceil}$}};
\draw[color=black] (7,4) node[scale =0.7]
{\rotatebox{0}{$x_1$}};
\draw[color=black] (7,-3) node[scale =0.7]
{\rotatebox{0}{$x_n$}};
\draw [fill=black] (3,4) circle (1.25pt);
\draw [fill=black] (3.5,4) circle (1.25pt);
\draw [fill=black] (4,4) circle (1.25pt);
\draw [fill=black] (3,-3) circle (1.25pt);
\draw [fill=black] (3.5,-3) circle (1.25pt);
\draw [fill=black] (4,-3) circle (1.25pt);
\end{tikzpicture}}
\vspace{-6mm}
\caption{The elimination tree with respect to $l_n^2$.}
\label{fig:ln2}
\end{minipage}}
\end{figure}
\vspace{-3mm}

Note that the height of the tree in Fig. \ref{fig:ln1} is $n-1$ and the height of the tree in Fig. \ref{fig:ln2} is $\left\lceil\frac{n+1}{2}\right\rceil$. Hence, from the complexity analysis in the last section, $l_n^2$ can be a better variable ordering.
We compute the CAD's for the polynomial set $F_n$ by the CAD tools in Mathematica 12 
and in Maple 2020. 
In Mathematica 12, we use the command ``CylindricalDecomposition'' \cite{strzebonski2006cylindrical} for the formula $\bigwedge_{f\in F_n}f\neq 0$ while in  Maple 2020, we use the command ``PartialCylindricalAlgebraicDecomposition'' \cite{yang2001complete} for the polynomial $\prod F_n$ instead. The results are shown in Table \ref{tab:F1}. We see that the runtime with respect to the ordering $l_n^2$ with a smaller elimination tree height is much shorter. Also, less sample points are obtained with respect to the ordering $l_n^2$ (For this example, these two tools give the same number of sample points in each case shown in the table). Both of these are well consistent with the previous complexity analysis.
\end{example}
         
  
\begin{table}[h]
\vspace{-2mm}
    \centering
      \caption{Comparing the runtime and the number of sample points  of the CAD computation via the perfect elimination orderings $l_n^1$ and $l_n^2$ with different elimination tree height (``OM'' means out of memory).}
      \vspace{-2.5mm}
       \scalebox{0.89}{
    \begin{tabular}{|c||c|c|c| c c c|}
     \hline
          $n$   &  & $l_n^1$ &$l_n^2$ &\multicolumn{3}{c|}{Random Orderings} \\
        \hline
\multirow{3}*{8}&Mathematica 12 &0.32s&0.17s& 18.13s &28.17s &1.87s\\
        \cline{2-7}
        &Maple 2020 &0.09s&0.04s&8.80s&5.97s&0.27s\\     
        \cline{2-7}
        &\#sample &576&256&11520&16000&1728\\     
         \hline
         
         \multirow{3}*{9}&Mathematica 12 &1.04s&0.33s&7.95s&4.66s &43.16s\\
        \cline{2-7}
        &Maple 2020 &0.28s&0.08s&1.32s&1.0s&24.86s\\   
         \cline{2-7}
        &\#sample &1728&512&5184&3456&23552\\  
         \hline
           \multirow{3}*{13}&Mathematica 12 &122s&22.4s& >2h &>2h &>2h\\
        \cline{2-7}
        &Maple 2020 &397s&29.1s&>2h&OM&OM\\  
        \cline{2-7}
        &\#sample &139968&41472&--&--&--\\  
         \hline
           \multirow{3}*{14}&Mathematica 12 &405s&74.3s&>2h &>2h &>2h\\
        \cline{2-7}
        &Maple 2020 &1744s&329s&>2h&OM&>2h\\     
         \cline{2-7}
        &\#sample &419904&124416&--&--&--\\  
         \hline
    \end{tabular}}
  
    \label{tab:F1}
    \vspace{-3mm}
\end{table}
\subsection{The Impact of Different Chordal Completions}
In practice, the associated graph of a sparse polynomial system is not necessarily a chordal  graph. In this case we need to compute one of its chordal completions together with a perfect elimination ordering of that completion. Conversely, from any variable ordering one can also construct a corresponding chordal completion via the elimination game. For the case where the associated graph is not chordal, we recommend those variable orderings for CAD computation, which, via the elimination game, result in minimal chordal completions. The following example shows the impact of different chordal completions on the CAD computation.

\begin{example}\hspace{-0.4mm}(\cite{cifuentes2017chordal})
Consider the following polynomial set:
\[
I^{n_1,n_2}=\bigcup_{0\leq i<n_1,0\leq j<n_2}
\left\{
\begin{array}{c} 
    U_{i,j}R_{i,j+1}-R_{i,j}U_{i+1,j},\\
    D_{i,j+1}R_{i,j}-R_{i,j+1}D_{i+1,j+1}, \\
    D_{i+1,j+1}L_{i+1,j}-L_{i+1,j+1}D_{i,j+1},\\
    U_{i+1,j}L_{i+1,j+1}-L_{i+1,j}U_{i,j}
\end{array}
\right\}.
\]

In Table \ref{tab:In1n2} we show the runtime and the numbers of sample points of the CAD computation for $I^{n_1,n_2}$. Therein, the symbol ``$n$'' denotes the number of variables and ``$l_p$'' denotes the variable orderings which, via the elimination game, result in minimal chordal completions of the associated graph of $I^{n_1,n_2}$. We know that each of these variable orderings is a perfect elimination ordering of the corresponding completion. The symbol ``$l_{sv}$'' denotes the variable orderings given by the Maple command ``SuggestVariableOrder'' while  ``$l_{R}$'' denotes some of the random variable orderings. The letter ``$t$'' denotes the runtime of the command ``PartialCylindricalAlgebraicDecomposition'' for the polynomial $\prod I^{n_1,n_2}$ in Maple 2020, while ``$m$'' denotes the number of sample points. The number $d$ is defined by $d=1-{|G|}/{|\bar{G}|}$ with $|G|$ the number of the edges of the  associated graph  and $|\bar{G}|$ the number of the edges of its chordal completion obtained by playing the elimination game via the corresponding variable ordering in the table. The number $d$ characterizes the difference between the associated graph and its chordal completion.

\begin{table}[htbp]
    \centering
    \caption{The CAD computation via different variable orderings resulting different  chordal completions.}
    \vspace{-2mm}
    \label{tab:In1n2}
    \begin{tabular}{|c|c|c|c|c|c|c|}
       \hline
         $(n_1 ,n_2)$&$n$& & $l_p$ &$l_{sv}$&\multicolumn{2}{c|}{$l_{R}$}\\
       \hline
       \multirow{3}*{(1,1)}& \multirow{3}*{8}& $t$&0.04s  & 0.14s  & 0.18s & 0.12s \\
       \cline{3-7}
      & & $d$& 0.17 & 0.17 & 0.23 &  0.20\\
      \cline{3-7}
        & & $m$& 288& 864 &1153 &  864\\
       \hline
       \multirow{3}*{(2,1)}&\multirow{3}*{14}&$t$& 361s  & 780s  &OM  & 1132s \\
       \cline{3-7}
      & & $d$& 0.17 &0.21 & 0.50 &  0.35\\
       \cline{3-7}
        & & $m$& 248832& 248832 &-- &  248832\\
     \hline
       \multirow{3}*{ (1,2)}&\multirow{3}*{14}&$t$& 581s & 712s  &1147s  & 2868s \\
       \cline{3-7}
      & & $d$& 0.17&0.14& 0.46 & 0.51 \\
       \cline{3-7}
        & & $m$& 248832& 248832 &248832 &  248832\\
     \hline

    \end{tabular}
\vspace{-2mm}
\end{table}

The data in Table \ref{tab:In1n2} show that, roughly speaking, the smaller the number $d$ is, the faster the computation may be done. In other words, a variable ordering could be better for the CAD projection if the difference between the corresponding chordal completion and the original associated graph is smaller. On the other hand, the relations between the number of sample points and $d$ seem not easy to grasp. One thing is clear that, for this example, the numbers of sample points obtained in the cases with smaller $d$ are less than or equal to the ones obtained in the other cases with bigger $d$. A somewhat surprising fact is that, when $n=14$, the number of sample points keeps unchanged under several different variable orderings.
\end{example}


\section{Conclusion}\label{sec:conclusion}
In this paper, we indicate for the first time that classical CAD projection algorithms preserve the chordal structure of a polynomial set, if the projection process is done in accordance with a perfect elimination ordering. More importantly, we provide a complexity analysis for the CAD algorithms executed with respect to a perfect elimination ordering by giving a new $(m,d)$-property which characterizes the ``size'' of the corresponding full set of projection polynomials. Combining these theoretical results with the experiments, we find that: (i) When the polynomial set is sparse and of the chordal structure, it can be better to compute the CAD via the perfect elimination orderings which result in perfect elimination trees of the minimal height. (ii) When the polynomial set is not of the chordal structure but still nearly chordal (and also sparse),  it could be better to use the variable orderings that result in minimal chordal completions of the original system. The core idea is that we use the perfect elimination ordering during the computation so that the sparsity embedded in the (nearly) chordal structure can pass from the original polynomial set to all the successive polynomial sets produced later on. Both the complexity analysis and the experimental results show that this idea really helps reduce the size of the full set of projection polynomials in many cases.

\vspace{-2mm}
\section*{Acknowledgments}
The work has been supported by the NSFC under grant No. 61732001.

\vspace{-2mm}
\bibliographystyle{ACM-Reference-Format}
\bibliography{chordalcad}


\begin{thebibliography}{32}


\ifx \showCODEN    \undefined \def \showCODEN     #1{\unskip}     \fi
\ifx \showDOI      \undefined \def \showDOI       #1{#1}\fi
\ifx \showISBNx    \undefined \def \showISBNx     #1{\unskip}     \fi
\ifx \showISBNxiii \undefined \def \showISBNxiii  #1{\unskip}     \fi
\ifx \showISSN     \undefined \def \showISSN      #1{\unskip}     \fi
\ifx \showLCCN     \undefined \def \showLCCN      #1{\unskip}     \fi
\ifx \shownote     \undefined \def \shownote      #1{#1}          \fi
\ifx \showarticletitle \undefined \def \showarticletitle #1{#1}   \fi
\ifx \showURL      \undefined \def \showURL       {\relax}        \fi
\providecommand\bibfield[2]{#2}
\providecommand\bibinfo[2]{#2}
\providecommand\natexlab[1]{#1}
\providecommand\showeprint[2][]{arXiv:#2}

\bibitem[\protect\citeauthoryear{Berry, Blair, Heggernes, and Peyton}{Berry
  et~al\mbox{.}}{2004}]%
        {DBLP:journals/algorithmica/BerryBHP04}
\bibfield{author}{\bibinfo{person}{A. Berry}, \bibinfo{person}{J.~R.~S. Blair},
  \bibinfo{person}{P. Heggernes}, {and} \bibinfo{person}{B.~W. Peyton}.}
  \bibinfo{year}{2004}\natexlab{}.
\newblock \showarticletitle{Maximum Cardinality Search for Computing Minimal
  Triangulations of Graphs}.
\newblock \bibinfo{journal}{\emph{Algorithmica}} \bibinfo{volume}{39},
  \bibinfo{number}{4} (\bibinfo{year}{2004}), \bibinfo{pages}{287--298}.
\newblock


\bibitem[\protect\citeauthoryear{Blair and Peyton}{Blair and Peyton}{1993}]%
        {Blair1993}
\bibfield{author}{\bibinfo{person}{J.~R.~S. Blair} {and} \bibinfo{person}{B.
  Peyton}.} \bibinfo{year}{1993}\natexlab{}.
\newblock \showarticletitle{An Introduction to Chordal Graphs and Clique
  Trees}. In \bibinfo{booktitle}{\emph{Graph Theory and Sparse Matrix
  Computation}}. \bibinfo{publisher}{Springer New York},
  \bibinfo{pages}{1--29}.
\newblock
\showISBNx{978-1-4613-8369-7}


\bibitem[\protect\citeauthoryear{Bradford, Davenport, England, and
  Wilson}{Bradford et~al\mbox{.}}{2013}]%
        {bradford2013optimising}
\bibfield{author}{\bibinfo{person}{R. Bradford}, \bibinfo{person}{J.H.
  Davenport}, \bibinfo{person}{M. England}, {and} \bibinfo{person}{D. Wilson}.}
  \bibinfo{year}{2013}\natexlab{}.
\newblock \showarticletitle{Optimising Problem Formulation for Cylindrical
  Algebraic Decomposition}. In \bibinfo{booktitle}{\emph{International
  Conference on Intelligent Computer Mathematics}}. Springer,
  \bibinfo{pages}{19--34}.
\newblock


\bibitem[\protect\citeauthoryear{Bradford, Davenport, England, McCallum, and
  Wilson}{Bradford et~al\mbox{.}}{2016}]%
        {DBLP:journals/jsc/BradfordDEMW16}
\bibfield{author}{\bibinfo{person}{R.~J. Bradford}, \bibinfo{person}{J.~H.
  Davenport}, \bibinfo{person}{M. England}, \bibinfo{person}{S. McCallum},
  {and} \bibinfo{person}{D.~J. Wilson}.} \bibinfo{year}{2016}\natexlab{}.
\newblock \showarticletitle{Truth Table Invariant Cylindrical Algebraic
  Decomposition}.
\newblock \bibinfo{journal}{\emph{Journal of Symbolic Computation}}
  \bibinfo{volume}{76} (\bibinfo{year}{2016}), \bibinfo{pages}{1--35}.
\newblock


\bibitem[\protect\citeauthoryear{Brown}{Brown}{2001}]%
        {DBLP:journals/jsc/Brown01a}
\bibfield{author}{\bibinfo{person}{C.~W. Brown}.}
  \bibinfo{year}{2001}\natexlab{}.
\newblock \showarticletitle{Improved Projection for Cylindrical Algebraic
  Decomposition}.
\newblock \bibinfo{journal}{\emph{Journal of Symbolic Computation}}
  \bibinfo{volume}{32}, \bibinfo{number}{5} (\bibinfo{year}{2001}),
  \bibinfo{pages}{447--465}.
\newblock


\bibitem[\protect\citeauthoryear{Brown}{Brown}{2004}]%
        {brown2004companion}
\bibfield{author}{\bibinfo{person}{C.~W. Brown}.}
  \bibinfo{year}{2004}\natexlab{}.
\newblock \showarticletitle{Companion to the Tutorial: Cylindrical Algebraic
  Decomposition}.
\newblock \bibinfo{journal}{\emph{Presented at ISSAC'04}}
  (\bibinfo{year}{2004}).
\newblock


\bibitem[\protect\citeauthoryear{Brown}{Brown}{2015}]%
        {brown2015open}
\bibfield{author}{\bibinfo{person}{C.~W. Brown}.}
  \bibinfo{year}{2015}\natexlab{}.
\newblock \showarticletitle{Open Non-Uniform Cylindrical Algebraic
  Decompositions}. In \bibinfo{booktitle}{\emph{Proc. ISSAC'15}}.
  \bibinfo{publisher}{ACM Press}, \bibinfo{pages}{85--92}.
\newblock


\bibitem[\protect\citeauthoryear{Chen}{Chen}{2020}]%
        {chen2020chordality}
\bibfield{author}{\bibinfo{person}{C. Chen}.} \bibinfo{year}{2020}\natexlab{}.
\newblock \showarticletitle{Chordality Preserving Incremental Triangular
  Decomposition and Its Implementation}. In
  \bibinfo{booktitle}{\emph{ICMS'20}}. Springer, \bibinfo{pages}{27--36}.
\newblock


\bibitem[\protect\citeauthoryear{Chen, Maza, Xia, and Yang}{Chen
  et~al\mbox{.}}{2009}]%
        {triCAD}
\bibfield{author}{\bibinfo{person}{C. Chen}, \bibinfo{person}{M.~Moreno Maza},
  \bibinfo{person}{B. Xia}, {and} \bibinfo{person}{L. Yang}.}
  \bibinfo{year}{2009}\natexlab{}.
\newblock \showarticletitle{Computing Cylindrical Algebraic Decomposition via
  Triangular Decomposition}. In \bibinfo{booktitle}{\emph{Proc. ISSAC'09}}.
  \bibinfo{publisher}{ACM Press}, \bibinfo{pages}{95--102}.
\newblock


\bibitem[\protect\citeauthoryear{Chen, Zhu, and Chi}{Chen
  et~al\mbox{.}}{2020}]%
        {chen2020variable}
\bibfield{author}{\bibinfo{person}{C. Chen}, \bibinfo{person}{Z. Zhu}, {and}
  \bibinfo{person}{H. Chi}.} \bibinfo{year}{2020}\natexlab{}.
\newblock \showarticletitle{Variable Ordering Selection for Cylindrical
  Algebraic Decomposition with Artificial Neural Networks}. In
  \bibinfo{booktitle}{\emph{International Congress on Mathematical Software}}.
  Springer, \bibinfo{pages}{281--291}.
\newblock


\bibitem[\protect\citeauthoryear{Cifuentes and Parrilo}{Cifuentes and
  Parrilo}{2016}]%
        {cifuentes2016exploiting}
\bibfield{author}{\bibinfo{person}{D. Cifuentes} {and} \bibinfo{person}{P.~A.
  Parrilo}.} \bibinfo{year}{2016}\natexlab{}.
\newblock \showarticletitle{Exploiting Chordal Structure in Polynomial Ideals:
  A Gr\"{o}bner Bases Approach}.
\newblock \bibinfo{journal}{\emph{SIAM Journal on Discrete Mathematics}}
  \bibinfo{volume}{30}, \bibinfo{number}{3} (\bibinfo{year}{2016}),
  \bibinfo{pages}{1534--1570}.
\newblock


\bibitem[\protect\citeauthoryear{Cifuentes and Parrilo}{Cifuentes and
  Parrilo}{2017}]%
        {cifuentes2017chordal}
\bibfield{author}{\bibinfo{person}{D. Cifuentes} {and} \bibinfo{person}{P.~A.
  Parrilo}.} \bibinfo{year}{2017}\natexlab{}.
\newblock \showarticletitle{Chordal Networks of Polynomial Ideals}.
\newblock \bibinfo{journal}{\emph{SIAM Journal on Applied Algebra and
  Geometry}} \bibinfo{volume}{1}, \bibinfo{number}{1} (\bibinfo{year}{2017}),
  \bibinfo{pages}{73--110}.
\newblock


\bibitem[\protect\citeauthoryear{Collins}{Collins}{1975}]%
        {DBLP:conf/automata/Collins75}
\bibfield{author}{\bibinfo{person}{G.~E. Collins}.}
  \bibinfo{year}{1975}\natexlab{}.
\newblock \showarticletitle{Quantifier Elimination for Real Closed Fields by
  Cylindrical Algebraic Decomposition}. In \bibinfo{booktitle}{\emph{Automata
  Theory and Formal Languages, 2nd {GI} Conference, Kaiserslautern, May 20-23,
  1975}} \emph{(\bibinfo{series}{Lecture Notes in Computer Science},
  Vol.~\bibinfo{volume}{33})}, \bibfield{editor}{\bibinfo{person}{H.~Barkhage}}
  (Ed.). \bibinfo{publisher}{Springer}, \bibinfo{pages}{134--183}.
\newblock


\bibitem[\protect\citeauthoryear{Collins and Hong}{Collins and Hong}{1991}]%
        {collins1991partial}
\bibfield{author}{\bibinfo{person}{G.~E. Collins} {and} \bibinfo{person}{H.
  Hong}.} \bibinfo{year}{1991}\natexlab{}.
\newblock \showarticletitle{Partial Cylindrical Algebraic Decomposition for
  Quantifier Elimination}.
\newblock \bibinfo{journal}{\emph{Journal of Symbolic Computation}}
  \bibinfo{volume}{12}, \bibinfo{number}{3} (\bibinfo{year}{1991}),
  \bibinfo{pages}{299--328}.
\newblock


\bibitem[\protect\citeauthoryear{Diaconis, Eisenbud, and Sturmfels}{Diaconis
  et~al\mbox{.}}{1998}]%
        {Diaconis1998}
\bibfield{author}{\bibinfo{person}{P. Diaconis}, \bibinfo{person}{D. Eisenbud},
  {and} \bibinfo{person}{B. Sturmfels}.} \bibinfo{year}{1998}\natexlab{}.
\newblock \bibinfo{booktitle}{\emph{Lattice Walks and Primary Decomposition}}.
\newblock \bibinfo{publisher}{Birkh{\"a}user Boston},
  \bibinfo{pages}{173--193}.
\newblock


\bibitem[\protect\citeauthoryear{Dolzmann, Seidl, and Sturm}{Dolzmann
  et~al\mbox{.}}{2004}]%
        {dolzmann2004efficient}
\bibfield{author}{\bibinfo{person}{A. Dolzmann}, \bibinfo{person}{A. Seidl},
  {and} \bibinfo{person}{T. Sturm}.} \bibinfo{year}{2004}\natexlab{}.
\newblock \showarticletitle{Efficient Projection Orders for CAD}. In
  \bibinfo{booktitle}{\emph{Proc. ISSAC'04}}. \bibinfo{publisher}{ACM Press},
  \bibinfo{pages}{111--118}.
\newblock


\bibitem[\protect\citeauthoryear{England and Florescu}{England and
  Florescu}{2019}]%
        {england2019comparing}
\bibfield{author}{\bibinfo{person}{M. England} {and} \bibinfo{person}{D.
  Florescu}.} \bibinfo{year}{2019}\natexlab{}.
\newblock \showarticletitle{Comparing Machine Learning Models to Choose the
  Variable Ordering for Cylindrical Algebraic Decomposition}. In
  \bibinfo{booktitle}{\emph{International Conference on Intelligent Computer
  Mathematics}}. Springer, \bibinfo{pages}{93--108}.
\newblock


\bibitem[\protect\citeauthoryear{Han, Dai, Hong, and Xia}{Han
  et~al\mbox{.}}{2017}]%
        {han2017open}
\bibfield{author}{\bibinfo{person}{J. Han}, \bibinfo{person}{L. Dai},
  \bibinfo{person}{H. Hong}, {and} \bibinfo{person}{B. Xia}.}
  \bibinfo{year}{2017}\natexlab{}.
\newblock \showarticletitle{Open Weak CAD and Its Applications}.
\newblock \bibinfo{journal}{\emph{Journal of Symbolic Computation}}
  \bibinfo{volume}{80} (\bibinfo{year}{2017}), \bibinfo{pages}{785--816}.
\newblock


\bibitem[\protect\citeauthoryear{Han, Jin, and Xia}{Han et~al\mbox{.}}{2016}]%
        {han2016proving}
\bibfield{author}{\bibinfo{person}{J. Han}, \bibinfo{person}{Z. Jin}, {and}
  \bibinfo{person}{B. Xia}.} \bibinfo{year}{2016}\natexlab{}.
\newblock \showarticletitle{Proving Inequalities and Solving Global
  Optimization Problems via Simplified CAD Projection}.
\newblock \bibinfo{journal}{\emph{Journal of Symbolic Computation}}
  \bibinfo{volume}{72} (\bibinfo{year}{2016}), \bibinfo{pages}{206--230}.
\newblock


\bibitem[\protect\citeauthoryear{Hong}{Hong}{1990a}]%
        {Hong}
\bibfield{author}{\bibinfo{person}{H. Hong}.} \bibinfo{year}{1990}\natexlab{a}.
\newblock \showarticletitle{An Improvement of the Projection Operator in
  Cylindrical Algebraic Decomposition}. In \bibinfo{booktitle}{\emph{Proc.
  ISSAC'1990}}. \bibinfo{publisher}{ACM Press}, \bibinfo{pages}{261--264}.
\newblock


\bibitem[\protect\citeauthoryear{Hong}{Hong}{1990b}]%
        {hong1990improvements}
\bibfield{author}{\bibinfo{person}{H. Hong}.} \bibinfo{year}{1990}\natexlab{b}.
\newblock \emph{\bibinfo{title}{Improvements in CAD-Based Quantifier
  Elimination}}.
\newblock \bibinfo{thesistype}{Ph.D. Dissertation}. \bibinfo{school}{The Ohio
  State University}.
\newblock


\bibitem[\protect\citeauthoryear{Hong and Din}{Hong and Din}{2012}]%
        {hong2012variant}
\bibfield{author}{\bibinfo{person}{H. Hong} {and} \bibinfo{person}{M.~Safey~El
  Din}.} \bibinfo{year}{2012}\natexlab{}.
\newblock \showarticletitle{Variant Quantifier Elimination}.
\newblock \bibinfo{journal}{\emph{Journal of Symbolic Computation}}
  \bibinfo{volume}{47}, \bibinfo{number}{7} (\bibinfo{year}{2012}),
  \bibinfo{pages}{883--901}.
\newblock


\bibitem[\protect\citeauthoryear{Huang, England, Wilson, Davenport, Paulson,
  and Bridge}{Huang et~al\mbox{.}}{2014}]%
        {huang2014applying}
\bibfield{author}{\bibinfo{person}{Z. Huang}, \bibinfo{person}{M. England},
  \bibinfo{person}{D. Wilson}, \bibinfo{person}{J.~H. Davenport},
  \bibinfo{person}{L.C. Paulson}, {and} \bibinfo{person}{J. Bridge}.}
  \bibinfo{year}{2014}\natexlab{}.
\newblock \showarticletitle{Applying Machine Learning to the Problem of
  Choosing a Heuristic to Select the Variable Ordering for Cylindrical
  Algebraic Decomposition}. In \bibinfo{booktitle}{\emph{International
  Conference on Intelligent Computer Mathematics}}. Springer,
  \bibinfo{pages}{92--107}.
\newblock


\bibitem[\protect\citeauthoryear{McCallum}{McCallum}{1985}]%
        {McCallum85}
\bibfield{author}{\bibinfo{person}{S. McCallum}.}
  \bibinfo{year}{1985}\natexlab{}.
\newblock \emph{\bibinfo{title}{An Improved Projection Operation for
  Cylindrical Algebraic Decomposition}}.
\newblock \bibinfo{thesistype}{Ph.D. Dissertation}. \bibinfo{school}{University
  of Wisconsin-Madison}.
\newblock


\bibitem[\protect\citeauthoryear{McCallum}{McCallum}{1988}]%
        {DBLP:journals/jsc/McCallum88}
\bibfield{author}{\bibinfo{person}{S. McCallum}.}
  \bibinfo{year}{1988}\natexlab{}.
\newblock \showarticletitle{An Improved Projection Operation for Cylindrical
  Algebraic Decomposition of Three-Dimensional Space}.
\newblock \bibinfo{journal}{\emph{Journal Symbolic Computation}}
  \bibinfo{volume}{5}, \bibinfo{number}{1/2} (\bibinfo{year}{1988}),
  \bibinfo{pages}{141--161}.
\newblock


\bibitem[\protect\citeauthoryear{McCallum}{McCallum}{1998}]%
        {McCallum2}
\bibfield{author}{\bibinfo{person}{S. McCallum}.}
  \bibinfo{year}{1998}\natexlab{}.
\newblock \showarticletitle{An Improved Projection Operation for Cylindrical
  Algebraic Decomposition}.
\newblock In \bibinfo{booktitle}{\emph{Quantifier Elimination and Cylindrical
  Algebraic Decomposition}}. \bibinfo{publisher}{Springer},
  \bibinfo{pages}{242--268}.
\newblock


\bibitem[\protect\citeauthoryear{Mou, Bai, and Lai}{Mou et~al\mbox{.}}{2019}]%
        {mou2019chordal}
\bibfield{author}{\bibinfo{person}{C. Mou}, \bibinfo{person}{Y. Bai}, {and}
  \bibinfo{person}{J. Lai}.} \bibinfo{year}{2019}\natexlab{}.
\newblock \showarticletitle{Chordal Graphs in Triangular Decomposition in
  Top-Down Style}.
\newblock \bibinfo{journal}{\emph{Journal of Symbolic Computation}}
  \bibinfo{volume}{102} (\bibinfo{year}{2019}), \bibinfo{pages}{108--131}.
\newblock


\bibitem[\protect\citeauthoryear{Mou and Lai}{Mou and Lai}{2019}]%
        {mou2019chordality}
\bibfield{author}{\bibinfo{person}{C. Mou} {and} \bibinfo{person}{J. Lai}.}
  \bibinfo{year}{2019}\natexlab{}.
\newblock \showarticletitle{On the Chordality of Simple Decomposition in
  Top-Down Style}. In \bibinfo{booktitle}{\emph{International Conference on
  Mathematical Aspects of Computer and Information Sciences}}. Springer,
  \bibinfo{pages}{138--152}.
\newblock


\bibitem[\protect\citeauthoryear{Parter}{Parter}{1961}]%
        {doi:10.1137/1003021}
\bibfield{author}{\bibinfo{person}{S. Parter}.}
  \bibinfo{year}{1961}\natexlab{}.
\newblock \showarticletitle{The Use of Linear Graphs in Gauss Elimination}.
\newblock \bibinfo{journal}{\emph{SIAM $\text{Review}$}} \bibinfo{volume}{3},
  \bibinfo{number}{2} (\bibinfo{year}{1961}), \bibinfo{pages}{119--130}.
\newblock


\bibitem[\protect\citeauthoryear{Strzebo{\'n}ski}{Strzebo{\'n}ski}{2006}]%
        {strzebonski2006cylindrical}
\bibfield{author}{\bibinfo{person}{A.~W. Strzebo{\'n}ski}.}
  \bibinfo{year}{2006}\natexlab{}.
\newblock \showarticletitle{Cylindrical Algebraic Decomposition Using Validated
  NumericS}.
\newblock \bibinfo{journal}{\emph{Journal of Symbolic Computation}}
  \bibinfo{volume}{41}, \bibinfo{number}{9} (\bibinfo{year}{2006}),
  \bibinfo{pages}{1021--1038}.
\newblock


\bibitem[\protect\citeauthoryear{Strzebo{\'n}ski}{Strzebo{\'n}ski}{2016}]%
        {strzebonski2016cylindrical}
\bibfield{author}{\bibinfo{person}{A.~W. Strzebo{\'n}ski}.}
  \bibinfo{year}{2016}\natexlab{}.
\newblock \showarticletitle{Cylindrical Algebraic Decomposition Using Local
  Projections}.
\newblock \bibinfo{journal}{\emph{Journal of Symbolic Computation}}
  \bibinfo{volume}{76} (\bibinfo{year}{2016}), \bibinfo{pages}{36--64}.
\newblock


\bibitem[\protect\citeauthoryear{Yang, Hou, and Xia}{Yang
  et~al\mbox{.}}{2001}]%
        {yang2001complete}
\bibfield{author}{\bibinfo{person}{L. Yang}, \bibinfo{person}{X. Hou}, {and}
  \bibinfo{person}{B. Xia}.} \bibinfo{year}{2001}\natexlab{}.
\newblock \showarticletitle{A Complete Algorithm for Automated Discovering of a
  Class of Inequality-Type Theorems}.
\newblock \bibinfo{journal}{\emph{Science in China Series F Information
  Sciences}} \bibinfo{volume}{44}, \bibinfo{number}{1} (\bibinfo{year}{2001}),
  \bibinfo{pages}{33--49}.
\newblock


\end{thebibliography}
\end{document}